\documentclass[conference]{IEEEtran}

\IEEEoverridecommandlockouts

\usepackage{amsfonts,amsmath,amssymb,bbm}
\usepackage{graphicx,epsfig,color}
\usepackage{algorithm,algorithmic}
\usepackage{comment}
\usepackage{epstopdf}
\usepackage{bm}
\usepackage[font=footnotesize]{caption}
\usepackage{subfigure}

\newtheorem{thm}{Theorem}

\newtheorem{lem}{Lemma}

\newtheorem{proof}{Proof}[section]
\newtheorem{rem}{Remark}
\newtheorem{conj}{Conjecture}

\newcommand{\ie}{{\em i.e., }}
\newcommand{\eg}{{\em e.g., }}
\newcommand{\Qm}{{\bf Q}}
\newcommand{\Am}{{\bf A}}
\newcommand{\fv}{{\bf f}}
\newcommand{\yv}{{\bf y}}
\newcommand{\wv}{{\bf w}}

\newcommand{\ztr}{{j}}

\begin{document}

\title{Optimal Virtual Network Service Placement/Distribution}
\title{Approximation Algorithms for the NFV Service Distribution Problem}

\author{
\IEEEauthorblockN{
Hao Feng\IEEEauthorrefmark{1}, Jaime Llorca\IEEEauthorrefmark{2}, Antonia M. Tulino\IEEEauthorrefmark{2}\IEEEauthorrefmark{3}, Danny Raz\IEEEauthorrefmark{2}, Andreas F. Molisch\IEEEauthorrefmark{1}
}
\IEEEauthorblockA{
\IEEEauthorrefmark{1}University of Southern California, CA, USA, Email: \{haofeng, molisch\}@usc.edu\\
\IEEEauthorrefmark{2}Nokia Bell Labs, NJ, USA, Email: \{jaime.llorca, a.tulino, danny.raz\}@nokia-bell-labs.com\\
\IEEEauthorrefmark{3}University of Naples Federico II, Italy, Email: \{antoniamaria.tulino\}@unina.it
}
}


\maketitle

\begin{abstract}
Distributed cloud networking builds on network functions virtualization (NFV) and software defined networking (SDN) to enable the
deployment of network services in the form of elastic virtual network functions (VNFs) instantiated over general purpose servers at distributed cloud locations. 
We address the design of fast approximation algorithms for the NFV service distribution problem (NSDP), whose goal is to determine the placement of VNFs, the routing of service flows, and the associated allocation of cloud and network resources that satisfy client demands with minimum cost.
We show that in the case of 
load-proportional costs, the resulting fractional NSDP can be formulated as a {\em multi-commodity-chain} flow problem on a cloud-augmented graph, and design a queue-length based algorithm, named QNSD, that provides an $O(\epsilon)$ approximation in time $O(1/\epsilon)$.
We then address the case 
in which resource costs are a function of the integer number of allocated resources and design a variation of QNSD 
that effectively pushes for flow consolidation into a limited number of active resources to minimize overall cloud network cost.
\end{abstract}

\pagenumbering{arabic}

\IEEEpeerreviewmaketitle


\section{Introduction}

Distributed cloud networking builds on network functions virtualization (NFV) and software defined networking (SDN) to enable the
deployment of network services in the form of elastic virtual network functions (VNFs) instantiated over general purpose servers at multiple cloud locations, and interconnected by a programmable network fabric that allows dynamically steering client flows \cite{nfv}-\cite{nfv_sdn}.
A \emph{cloud network} operator can then host a variety of services over a common physical infrastructure, reducing both capital and operational expenses. 
In order to make the most of this attractive scenario, a key challenge is to find the placement of VNFs and the routing of client flows through the appropriate VNFs that minimize the use of the physical infrastructure. 

The work in \cite{vnf} first addressed the problem of placing VNFs in a capacitated cloud network. The authors formulated the problem as a generalization of Facility Location and Generalized Assignment, and provided near optimal solutions with bi-criteria constant approximations guarantees.
However, the model in \cite{vnf} does not account for function ordering or {\em service chaining} nor flow routing optimization.
Subsequently, the work in \cite{csdp} introduced a flow based model that allows optimizing the distribution (function placement and flow routing) of services with arbitrary function relationships (\eg chaining) over capacitated cloud networks. Cloud services are described via a directed acyclic graph and the function placement and flow routing is determined by solving a minimum cost network flow problem 
with service chaining constraints.
In \cite{csdp}, it is shown that the cloud service distribution problem (CSDP) admits optimal polynomial-time solutions under convex cost functions, but remains hard otherwise.
The model in \cite{csdp} also allows accounting for multicast flows, in which case, general polynomial solvability requires a convex cost function and the ability to perform intra-session network coding.

In this paper, we address the design of fast approximation algorithms for the NFV Service Distribution Problem (NSDP), a specially relevant class of the CSDP in which the service graph is a line graph (\ie chain) and all flows are unicast. 

Our contributions can be summarized as follows:

\begin{itemize}
\item We formulate the NSDP as a minimum cost {\em multi-commodity-chain} network design (MCCND) problem on a \emph{cloud-augmented graph}, where the goal is to find the placement of service functions, the routing of client flows through the appropriate service functions, and the corresponding allocation of cloud and network resources that minimize the cloud network operational cost. 
\item We first address the case of load-proportional resource costs. 
We show that the fractional NSDP becomes a min-cost multi-commodity-chain flow (MCCF) problem that admits optimal polynomial-time solutions. We design a queue-length based algorithm, named QNSD, that is shown to provide an $O(\epsilon)$ approximation to the fractional NSDP in time $O(1/\epsilon)$. We further conjecture that QNSD exhibits an improved $O(1/\sqrt\epsilon)$ convergence -- a result that, while not formally proved in this paper, is illustrated via simulations.
\item We then address the case in which resource costs are a function of the integer number of allocated resources. 
We design a new algorithm, C-QNSD,
which constrains the evolution of QNSD
to effectively drive service flows to consolidate on a limited number of active resources, yielding good practical solutions to the integer NSDP.
\end{itemize}

The rest of the paper is organized as follows. We review related work in Section \ref{related}. Section \ref{model} introduces the system model. Section \ref{problem} describes the network flow based formulation for the NSDP.
Sections \ref{linear} and \ref{integer} present the proposed approximation algorithms for the fractional and integer NSDP, respectively.
We present simulations results in Section \ref{sim}, discuss possible extensions in Section \ref{discussion}, and conclude the paper in Section \ref{conclusions}.

\section{Related Work}
\label{related}

To the best of our knowledge, the algorithms presented in this paper are the first approximation algorithms for the NSDP. While the NSDP can be seen as a special case of the CSDP introduced in \cite{csdp}, the authors only provided a network flow based formulation, without addressing the design of efficient approximation algorithms.

As shown in Section \ref{linear}, the fractional NSDP is a generalization of the minimum cost multi-commodity flow (MCF) problem. A large body of work has addressed the design of fast fully polynomial time approximation schemes (FPTAS) for MCF. The work in \cite{Garg07} summarizes the best known FPTAS for MCF and fractional packing problems. Based on \cite{Garg07}, the fastest schemes use shortest-path computations in each iteration to provide $O(\epsilon)$ approximations in time $O(1/\epsilon^2)$.
The scheme in \cite{Bienstock06} runs in time $O(1/\epsilon)$, but requires solving a convex quadratic program in each iteration, yielding slower running times for moderately small $\epsilon$. Our proposed QNSD algorithm for the fractional NSDP is shown to provide an $O(\epsilon)$ approximation in time $O(1/\epsilon)$, while simply solving a set of linear time max-weight problems in each iteration. We further conjecture that the running time of our algorithm is in fact $O(1/\sqrt{\epsilon})$. QNSD is hence also an improved FPTAS for MCF.
In \cite{Cao14}, the authors extended the shortest-path based algorithms in \cite{Garg07} to design policy-aware routing algorithms that steer network flows through pre-defined sequences of network functions in order to maximize the total served flow. The problem addressed in \cite{Garg07} can be thought of as a maximum flow version of the fractional NSDP, but their algorithms still run in time $O(1/\epsilon^2)$.

With respect to the integer NSDP, we show in Section \ref{integer} that it is a generalization of the well known NP-hard multi-commodity network design (MCND) problem \cite{bundle}. Our proposed C-QNSD algorithm extends QNSD with knapsack relaxation techniques similar to those used in MCND \cite{bundle}.

\section{System Model}
\label{model}

\subsection{Cloud network model}

We consider a cloud network modeled as a directed graph $\mathcal G=(\mathcal V,\mathcal E)$ with $n=|\mathcal V|$ vertices and $m=|\mathcal E|$ edges representing the set of nodes and links, respectively. 
A cloud network node represents a distributed cloud location, in which virtual network functions or VNFs can be instantiated in the form of \emph{e.g.} virtual machines (VMs) over general purpose servers \cite{nfv}.
When service flows go through VNFs at a cloud node, they consume cloud resources (\eg cpu, memory).
We denote by $w_u$ the cost per cloud resource unit (\eg server) at node $u$ and by $c_u$ the maximum number of cloud resource units that can be allocated at node $u$.
A cloud network link represents a network connection between two cloud locations.
When service flows go through a cloud network link, they consume network resources (\eg bandwidth).
We denote by $w_{uv}$ the cost per network resource unit (\eg 1 Gbps link) on link $(u,v)$ and by $c_{uv}$ the maximum number of network resource units that can be allocated on link $(u,v)$.

\subsection{Service model}

A service $\phi\in\Phi$ is described by a chain of $M_{\phi}$ VNFs.
We use the pair $(\phi,i)$, with $\phi\in\Phi, i\in\{1,\dots,M_{\phi}\}$, to denote the $i$-th function of service $\phi$, and $L$ to denote the to total number of available VNFs.
Each VNF is characterized by its cloud resource requirement, which may also depend on the specific cloud location. 
We denote by $r^{(\phi,i)}_u$ the cloud resource requirement (in cloud resource units per flow unit) of function $(\phi,i)$ at cloud node $u$.
That is, when one flow unit goes through function $(\phi,i)$ at cloud node $u$, it consumes $r^{(\phi,i)}_u$ cloud resource units.
In addition, when one flow unit goes through the fundamental transport function of link $(u,v)$, it consumes $r^{tr}_{uv}$ network resource units.



A client requesting service $\phi\in\Phi$ is represented by a destination node $d\in\mathcal D(\phi)\subset\mathcal V$, where $\mathcal D(\phi)$ denotes the set of clients requesting service $\phi$.
The demand of client $d$ for service $\phi$ is described by a set of source nodes $\mathcal S(d,\phi)\in\mathcal V$ and demands $\lambda^{d,\phi}_s$, $\forall s\in\mathcal S(d,\phi)$, indicating that a set of source flows, each of size $\lambda^{d,\phi}_s$ flow units and entering the network at $s\in\mathcal S(d,\phi)$, must go through the sequence of VNFs of service $\phi$ before exiting the network at destination node $d\in\mathcal D(\phi)$.
We set $\lambda_u^{(d,\phi)}=0, \forall u\notin\mathcal S(d,\phi)$ to indicate that only nodes in $\mathcal S(d,\phi)$ have source flows for the request of client $d$ for service $\phi$. 
We note that the adopted destination-based client model allows the total number of clients to scale linearly with the size of the network, as opposed to the quadratic scaling of the source-destination client model.

\section{The NFV Service Distribution Problem}
\label{problem}

\begin{figure}
\centering
\includegraphics[width=3.5in]{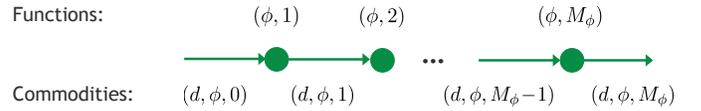}
\caption{Multi-commodity-chain flow model. Service $\phi$ for destination $d$ takes source commodity $(d,\phi,0)$ and processes it via function $(\phi,1)$ to create commodity $(d,\phi,1)$, which is then processed by function $(\phi,2)$... 
until 
function $(\phi,M_{\phi})$ 
produces final commodity $(d,\phi,M_{\phi})$.}
\label{flow_model}
\end{figure}

\begin{figure}
\centering
\includegraphics[width=1.3in]{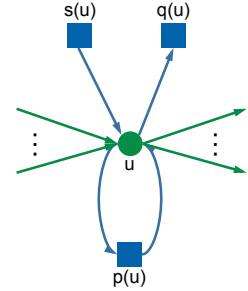}
\caption{Cloud-augmented graph for node $u$, where $p(u)$ represents the processing unit that hosts flow processing functions, $s(u)$ the source unit from which flows enter the cloud network, and $q(u)$ the demand unit via which flows exit the cloud network.}
\label{aug_graph}
\end{figure}


The goal of the NFV service distribution problem (NSDP) is to find the placement of service functions, the routing of service flows, and the associated allocation of cloud and network resources, that meet client demands  with minimum overall resource cost. In the following, we show how the NSDP can be solved by computing a {\em chained} network flow on a properly constructed graph.


We adopt a {\em multi-commodity-chain flow} (MCCF) model, in which a commodity is uniquely identified by the triplet $(d,\phi,i)$, which indicates that commodity $(d,\phi,i)$ is the output of the $i$-th function of service $\phi$ for client $d$ (see Fig. \ref{flow_model}).

We formulate the NSDP as a MCCF problem on the cloud-augmented graph that results from augmenting each node in $\mathcal G$ with the gadget in Fig. \ref{aug_graph},
where $p(u), s(u)$, and $q(u)$ denote the processing unit, source unit, and demand unit at cloud node $u$, respectively.
The resulting graph is denoted by $\mathcal G^a=(\mathcal V^a,\mathcal E^a)$, with $\mathcal V^a=\mathcal V\cup\mathcal V'$, $\mathcal E^a=\mathcal E\cup\mathcal E'$, and where $\mathcal V'$ and $\mathcal E'$ denote the set of processing, source, and demand unit nodes and edges, respectively.
We denote by $\delta^-(u)$ and $\delta^+(u)$ the set of incoming and outgoing neighbors of node $u$ in $\mathcal G^a$.

In the cloud-augmented graph $\mathcal G^a$, each edge $(u,v)\in\mathcal E^a$ has an associated capacity, unit resource cost, and per-function resource requirement, as described in the following.

For the set of edges $\{(u,p(u)), (p(u),u) : u\in\mathcal V\}$ representing the set of compute resources, we have:
\begin{itemize}
\item $c_{u,p(u)} = c_{u}$, $c_{p(u),u} = c_u^{max}$
\item $w_{u,p(u)} = w_u$, $w_{p(u),u} = 0$
\item $r^{(\phi,i)}_{u,p(u)} = r^{(\phi,i)}_u, \forall (\phi,i)$
\end{itemize}
where $c_u^{max}=\sum_{(d,\phi,i)}\sum_{s\in\mathcal S(d,\phi)} \lambda_s^{(d,\phi)} r_u^{(\phi,i)}$.
Note that we model the processing of network flows and associated allocation of compute resources using link $(u,p(u))$, and let edge $(p(u),u)$ be a free-cost edge of sufficiently high capacity that carries the processed flow back to node $u$.

For the set of edges $\{(s(u),u), (u,s(u)) : u\in\mathcal V\}$ representing the resources via which client flows enter and exit the cloud network, we have:
\begin{itemize}
\item $c_{s(u),u} = c_u^{max}$, $c_{u,q(u)} = c_u^{max}$
\item $w_{s(u),u} = 0$, $w_{u,q(u)} = 0$
\item $r^{(\phi,i)}_{s(u),u} = r^{(\phi,i)}_{u,q(u)} = 0, \forall (\phi,i)$
\end{itemize}
Note that we model the ingress and egress of network flows via free-cost edges of sufficiently high capacity. In addition, the irrelevant per-function resource requirement for these edges is set to zero.


Given that the set of network resources only perform the fundamental transport function of moving bits between cloud network nodes, the per-function resource requirement for the set of edges in the original graph $(u,v)\in\mathcal E$ is set to $r^{tr}_{uv}$ for all $(\phi,i)$, \ie $r^{(\phi,i)}_{uv} = r^{tr}_{uv}, \forall (\phi,i)$. The capacity and unit resource cost of network edge $(u,v)\in\mathcal E$ is given by $c_{uv}$ and $w_{uv}$, respectively.

We now define the following MCCF flow and resource allocation variables on the cloud-augmented graph $\mathcal G^a$:
\begin{itemize}
\item \emph{Flow variables:} 
$f^{(d,\phi,i)}_{uv}$ indicates the fraction of commodity $(d,\phi,i)$ on edge $(u,v)\in\mathcal E^a$, \ie the fraction of flow output of function $(\phi,i)$ for destination $d$ carried/processed by edge $(u,v)$.
\item \emph{Resource variables:} $y_{uv}$ indicates the total amount of resource units (\eg cloud or network resource units) allocated on edge $(u,v)\in\mathcal E^a$.
\end{itemize}


The NSDP can then be formulated via the following compact linear program:
\begin{subequations}\label{nsdp}
\begin{align}
&\text{min} &&
\sum_{(u,v)\in\mathcal E^a} w_{uv} \, y_{uv}    \label{obj}\\
&\text{s.t.} && \sum_{v\in\delta^{\!-\!}(u)} \!\! f_{vu}^{(d,\phi,i)} = \!\! \sum_{v\in\delta^{\!+\!}(u)} \!\! f_{uv}^{(d,\phi,i)} \qquad\quad\,\, \forall u, d, \phi, i \label{flowcons}\\
&&& f_{p(u),u}^{(d,\phi,i)} = f_{u,p(u)}^{(d,\phi,i-1)} \qquad\qquad\quad\quad \forall u, d, \phi, i \! \neq \! 0 \label{chain}\\
&&& \sum_{(d,\phi,i)} \! f_{uv}^{(d,\phi,i)} r^{(\phi,i+1)}_{uv} \! \leq y_{uv} \! \leq  c_{uv} \qquad\quad \forall (u,v) \label{cap}\\
&&& f_{s(u),u}^{(d,\phi,0)}=\lambda^{(d,\phi)}_{u}  \qquad\qquad\qquad\qquad\quad\,\,\, \forall u, d, \phi \label{source}\\
&&& f_{u,q(u)}^{(d,\phi,M_{\phi})}=0 \qquad\qquad\qquad\qquad\quad\, d, \phi, u\neq d \label{dest}\\
&&& f_{uv}^{(d,\phi,i)} \geq 0, \, y_{uv}\in\mathbb Z^+ \qquad\qquad\,\, \forall (u,v), d, \phi, i  \label{vars}
\end{align}
\end{subequations}
where, when not specified for compactness, $u\in\mathcal V$, $d\in\mathcal V$, $\phi\in\Phi$, $i\in\{1,\dots,M_{\phi}\}$, and $(u,v)\in\mathcal E^a$.



The objective is to minimize the overall cloud network resource cost, described by \eqref{obj}. Recall that the set $\mathcal E^a$ contains all edges in the augmented graph $\mathcal G^a$, representing both cloud and network resources.
Eq. \eqref{flowcons} describes standard flow conservation constraints applied to all nodes in $\mathcal V$.
A specially critical set of constraints are the {\em service chaining} constraints described by \eqref{chain}. These constraints establish that 
in order to have flow of a given commodity $(d,\phi,i)$ coming out of a processing unit, the input commodity $(d,\phi,i-1)$ must be entering the processing unit.
Constraints \eqref{cap} make sure that the total flow at a given cloud network resource is covered by enough resource units without violating capacity.\footnote{Recall from the MCCF model that commodity $(d,\phi,i)$ gets processed by function $(\phi,i+1)$, and that $r^{(\phi,i)}_{uv}=r^{tr}_{uv}, \forall (u,v)\in\mathcal E,\phi,i$.}
Eqs. \eqref{source} and \eqref{dest} establish the source and demand constraints.
Note from \eqref{dest} that no flows of final commodity $(d,\phi,M_{\phi})$ are allowed to exit the network other than at the destination node $d$.
Finally, \eqref{vars} describe the fractional and integer nature of flow and resource allocation variables, respectively.

In this work, we always assume fractional flow variables in order to capture the ability to split client flows to improve overall cloud network resource utilization.
With respect to the resource allocation variables, however,
we are interested in the following 
two main versions of the NSDP: 
\begin{itemize}
\item {\em Integer NSDP}:
The use of 
integer resource variables 
allows 
accurately capturing the allocation of an integer number of general purpose resource units (\eg servers).
In this case, the NSDP becomes a generalization of multi-commodity network design (MCND), where the network is augmented with {\em compute edges} that model the processing of service flows and where there are additional service chaining constraints that make sure flows follow service functions in the appropriate order.
In fact, for the special case that each service is composed of a single commodity, the NSDP is equivalent to the MCND.
We refer to the resulting MCCND problem as the integer NSDP. 
\item {\em Fractional NSDP}:
The use of 
fractional resource variables becomes a specially 
good compromise between accuracy and tractability
when the size of the cloud network resource units is much smaller than the total flow served at a given location. 
This is indeed the case for services that serve a large number of large-size flows (\eg telco services \cite{nfv}) and/or services deployed using small-size resource units (\eg micro-services \cite{micro}).
In this case, the NSDP becomes a generalization of min-cost MCF, with the exact equivalence holding in the case of single commodity services.
We refer to the resulting MCCF problem as the fractional NSDP.
\end{itemize}



\section{Fractional NSDP}
\label{linear}

As discussed in the previous section, the fractional NSDP can be formulated as the MCCF problem that results from the linear relaxation of \eqref{nsdp}, \ie by replacing $y_{uv}\in\mathbb Z^+$ with $y_{uv}\geq 0$ in \eqref{vars}.
While the resulting linear programming formulation admits optimal polynomial-time solutions, it requires solving a linear program with a large number of constraints.
In this work, we are interested in designing a fast fully polynomial approximation scheme (FPTAS) for the fractional NSDP.
A natural direction 
would be to build on state of the art FPTAS for MCF that rely on shortest-path computations \cite{Garg07,Cao14}. However, these techniques have only been shown to provide $O(\epsilon)$ approximations in time $O(1/\epsilon^2)$. Here, we target the design of faster approximations with order improvements in running time, \ie $O(1/\epsilon)$ and $O(1/\sqrt\epsilon)$, by redesigning queue-length based algorithms that have been shown to be very effective for stochastic network optimization, in order to construct fast iterative algorithms for static MCCF problems such as the fractional NSDP.

\subsection{QNSD algorithnm}
\label{QNSD_algorithnm}

In the following, we describe the proposed queue-length based network service distribution (QNSD) algorithm.
QNSD is
an iterative algorithm that mimics the time evolution of an underlying cloud network queueing system.
QNSD exhibits the following main key features:
\begin{itemize}
\item QNSD builds on a dynamic cloud network control algorithm recently introduced in \cite{dcnc}, which was designed to stabilize a stochastic cloud network system while minimizing the time-average resource cost, and uses the average over the algorithm iterations to compute the solution to the fractional NSDP.
\item Inspired by a recent result that characterizes the transient and steady state phases in queue-length based algorithms \cite{Neely14}, QNSD computes the solution to the fractional NSDP by averaging over a limited iteration horizon, yielding an $O(\epsilon)$ approximation in time $O(1/\epsilon)$. 
In the algorithm description, we use $j\in\mathbb Z^+$ to index the \emph{iteration frame} over which averages are computed.
\item QNSD further exploits the speed-up shown in gradient methods for convex optimization
when combining gradient directions over consecutive iterations \cite{Tseng08}, \cite{heavy_ball},
leading to a conjectured $O(1/\sqrt{\epsilon})$ convergence.
\end{itemize}

Before describing the algorithm, we shall introduce the following queue 
variables:
\begin{itemize}
\item \emph{Actual queues}:
$Q_u^{(d,\phi,i)}(t)$ denotes the queue backlog of commodity $(d,\phi,i)$ at node $u\in\mathcal V$ in iteration $t$. 
These queues represent the actual packet build-up that would take place in an equivalent dynamic cloud network system in which iterations correspond to time instants.
\item \emph{Virtual queues}:
$U_u^{(d,\phi,i)}(t)$ denotes the virtual queue backlog of commodity $(d,\phi,i)$ at node $u\in\mathcal V$ in iteration $t$. 
These virtual queues are used to capture the \emph{momentum} generated when combining differential queue backlogs (acting as gradients in our algorithm) over consecutive iterations \cite{heavy_ball}. 
\end{itemize}

\hfill

The QNSD algorithm works as follows:

\subsubsection{\underline{Initialization}}
\begin{align}
& f_{uv}^{(d,\phi,i)}(0) = y_{uv}(0) = 0 \qquad\qquad\qquad\qquad \forall (u,v) \notag \notag\\
& Q_u^{(d,\phi,i)}(0) = 0 \qquad\qquad\qquad\qquad\quad \forall (u,v), d, \phi, i \notag\\
& Q_d^{(d,\phi,M_\phi)}(t) = 0 \qquad\qquad\qquad\qquad\qquad\quad \forall d, \phi, t \notag\\
& U_u^{(d,\phi,i)}(0) =  U_u^{(d,\phi,i)}(-1) = 0 \qquad\qquad\, \forall u, d, \phi, i \notag\\
& U_d^{(d,\phi,M_{\phi})}(t) = 0 \qquad\qquad\qquad\qquad\qquad\, \forall d, \phi, i, t \notag\\
& f_{s(u),u}^{(d,\phi,i)}(t) =  \left \{
\begin{array}{lll}
 \lambda^{(d,\phi)}_u  &\qquad \forall u  \in\mathcal S(d,\phi),i=0, t \notag\\
 0 &\qquad   \mbox{\rm otherwise}
\end{array}
\right. \notag\\
& j = 0 \notag
\end{align}
Note that the queues associated with the final commodities at their respective destinations are set to zero for all $t$ to model the egress of flows from the cloud network.

%

\subsubsection{\underline{Main Procedure}}
In each iteration $t > 0$:

\begin{itemize}

\item \emph{Queue updates:}
For all $(d,\phi,i) \neq (u,\phi,M_{\phi})$:
\begin{align}
& \!\!\! Q_u^{(d,\phi,i)}(t) =  \left[ Q_u^{(d,\phi,i)}(t\!-\!1) - \! \sum_{v\in\delta^{\!+\!}(u)} \!\! f_{uv}^{(d,\phi,i)}(t\!-\!1) \right.  \notag\\
&\qquad\qquad\quad \left. + \! \sum_{v\in\delta^{\!-\!}(u)} \!\! f_{vu}^{(d,\phi,i)}(t\!-\!1) 
\right]^+ \label{queue_dynamics}\\
& \!\!\! \Delta Q_{uv}^{(d,\phi,i)}(t) = Q_u^{(d,\phi,i)}(t) - Q_u^{(d,\phi,i)}(t\!-\!1) \\
& \!\!\! U_{u}^{(d,\phi,i)}(t) = U_{u}^{(d,\phi,i)}(t\!-\!1) + \Delta Q_{uv}^{(d,\phi,i)}(t) \notag\\
&\qquad\qquad + \theta \left( U_{u}^{(d,\phi,i)}(t\!-\!1) - U_{u}^{(d,\phi,i)}(t\!-\!2) \right) \label{virtual_queue}
\end{align}
where $\theta\in[0,1)$ is a control parameter that
drives the differential queue backlog momentum.
Note that actual queues are updated according to standard queuing dynamics, while virtual queues are updated as a combination of the actual differential queue backlog and the virtual differential queue backlog in the previous iteration.

\item \emph{Transport decisions:}
For each link $(u,v)\in\mathcal E$:
\begin{itemize}

\item Compute the \emph{transport utility weight} of each commodity $(d,\phi,i)$:
\begin{align}
W_{uv}^{(d,\phi,i)}(t) = \frac{1}{r_{uv}^{tr}} \left( U_u^{(d,\phi,i)}(t) - U_v^{(d,\phi,i)}(t) \right) \notag
\end{align}
where $V$ is a control parameter that governs the optimality v.s. running-time tradeoff.

\item Compute the max-weight commodity $(d,\phi,i)^*$:
\begin{align}
(d,\phi,i)^* \! =  \underset{(d,\phi,i)}{\arg\max} \left\{ W_{uv}^{(d,\phi,i)}(t) \right\} \notag
\end{align}


\item Allocate network resources and assign flow rates
\begin{align}
& y_{uv}(t) =
\left \{
\begin{array}{lll}
c_{uv}  &\qquad  \text{if } W_{uv}^{(d,\phi,i)}(t) - V w_{uv}> 0 \notag\\
0 &\qquad   \mbox{\rm otherwise} \notag\\
\end{array}
\right. \notag\\
& f_{uv}^{(d,\phi,i)^*} = y_{uv}/r_{uv}^{tr} \notag\\
& f_{uv}^{(d,\phi,i)}(t) = 0, \quad \forall (d,\phi,i) \neq (d,\phi,i)^* \notag
\end{align}

\end{itemize}

\item \emph{Processing decisions:}
For each node $u\in\mathcal V$:

\begin{itemize}

\item Compute the \emph{processing utility weight} of each commodity $(d,\phi,i)$:
\begin{align}
& W_{u}^{(d,\phi,i)}(t)  = \frac{1}{r_u^{(\phi,i\!+\!1)}} \left( U_u^{(d,\phi,i)}(t) - U_u^{(d,\phi,i+1)}(t) \right) \notag
\end{align}
This key step in the QNSD computes the benefit of processing commodity $(d,\phi,i)$ via function $(\phi,i\!+\!1)$ at node $u$ in iteration $t$, taking into account the difference between the (virtual) queue backlog of commodity $(d,\phi,i)$ and that of the next commodity in the service chain $(d,\phi,i\!+\!1)$. Note also how a high cloud resource requirement $r_u^{(\phi,i\!+\!1)}$ reduces the benefit of processing commodity $(d,\phi,i)$.

\item Compute max-weight commodity $(d,\phi,i)^*$:
\begin{align}
(d,\phi,i)^*  =  \underset{(d,\phi,i)}{\arg\max} \left\{ W_{uv}^{(d,\phi,i)}(t) \right\} \notag
\end{align}


\item Allocate cloud resources and assign flow rates
\begin{align}
& y_{u,p(u)}(t) =
\left \{
\begin{array}{lll}
c_{u}  &\quad  \text{if } W_{uv}^{(d,\phi,i)}(t) - V w_{uv} > 0 \notag\\
0 &\quad   \mbox{\rm otherwise} \notag\\
\end{array}
\right. \notag\\
& y_{p(u),u} = y_{u,p(u)} \notag\\
& f_{u,p(u)}^{(d,\phi,i)^*}(t) = y_{u,p(u)}/r_{u}^{(\phi,i\!+\!1)} \notag\\
& f_{p(u),u}^{(d,\phi,i+1)^*}(t) = f_{u,p(u)}^{(d,\phi,i)^*}(t) \notag\\
& f_{u,p(u)}^{(d,\phi,i)}(t) = f_{p(u),u}^{(d,\phi,i)}(t) = 0 \quad \forall (d,\phi,i) \neq (d,\phi,i)^* \notag
\end{align}

\end{itemize}

\item \emph{Solution construction:}
\begin{itemize}
\item If \ $t=2^{j}$, \ then\ $t_{start} = t$ and $j=j+1$ \notag

\item Flow solution
\begin{align}
& \bar f_{uv}^{(d,\phi,i)} = \frac{{\sum_{\tau=t_{start}}^{t} f_{uv}^{(d,\phi,i)} (\tau)}}{t-t_{start}+1}  \quad\,\,\,\, \forall (u,v), d, \phi, i
\end{align}

\item Resource allocation solution
\begin{align}
& \bar y_{uv} = \frac{{\sum_{\tau=t_{start}}^{t} y_{uv}(\tau)}}{t-t_{start}+1}  \qquad\qquad\qquad \forall (u,v)
\end{align}

\end{itemize}


\end{itemize}

\hfill $\square$


\begin{rem}\label{rem1}
Note that QNSD solves $m+n$ max-weight problems in each iteration, leading to a running-time per iteration of $O((m+n)nL)=O(mnL)$.\footnote{Recall that the number of clients scales as $O(n)$ and $L$ is the total number of functions. Hence, the number of commodities scales as $O(nL)$.}
\end{rem}

\subsection{Performance of QNSD}

\begin{thm}
\label{thm1}
\emph{Let the input service demand ${\bm\lambda}=\{\lambda_u^{(d,\phi)}\}$ be such that the fractional NSDP is feasible and the Slater condition is satisfied. 
Then, letting $V=1/\epsilon$,
the QNSD algorithm provides an $O(\epsilon)$ approximation to the fractional NSDP in time $O(1/\epsilon)$.
Specifically, for all $t \ge T(\epsilon)$,
the QNSD solution $\{\bar f_{uv}^{(d,\phi,i)}, \bar y_{uv}\}$
satisfies: }
\begin{align}
& \! \sum_{(u,v)\in\mathcal E^a} \! w_{uv} \, \bar y_{uv}  \leq  h^{ \sf opt} + O(\epsilon) \label{eq:thm1}\\
& \! \sum_{v\in\delta^{\!-\!}(u)} \! \bar f_{vu}^{(d,\phi,i)} - \! \sum_{v\in\delta^{\!+\!}(u)} \! \bar f_{uv}^{(d,\phi,i)} \leq O(\epsilon)  \qquad\,\,\, \forall u, d, \phi, i \label{eq:thm2} \\
& \eqref{chain}-\eqref{vars}
\end{align}
\emph{where $h^{\sf opt}$ denotes the optimal objective function value and $T(\epsilon)$ is an $O(1/\epsilon)$ function, whose expression is derived in the theorem's proof given in Section \ref{Appedix:DCNC}. }
 \end{thm}

\begin{proof}
See Appendix in Section \ref{Appedix:DCNC}.
\end{proof}

\hfill

\begin{rem}\label{rem2}
While the claim of Theorem \ref{thm1} does not specify the dependence of the approximation on the size of the cloud network ($m,n$),  
in Section \ref{Appedix:DCNC}, we show that, in time $O(m/\epsilon)$, the total cost approximation is $O(m\epsilon)$ and the flow conservation violation is $O(\epsilon)$.
%
%
The total running time of QNSD is then $O(\epsilon^{-1}m^2nL)$. 
\end{rem}

\hfill

\begin{conj}\label{conj1}
\emph{With a properly chosen $\theta\in[0,1)$, QNSD provides an $O(\epsilon)$ approximation to the fractional NSDP in time $O(1/\sqrt{\epsilon})$.}
\end{conj}

\hfill $\blacksquare$

Our conjecture is based on the fact that: {\em i}) as shown in  the proof of Theorem \ref{thm1}, $\theta=0$ is sufficient for QNSD to achieve $O(1/\epsilon)$ convergence,  {\em ii}) recent results have shown $O(1/\sqrt{\epsilon})$ convergence of queue-length based algorithms for stochastic optimization when including a first-order memory or momentum of the differential queue backlog \cite{heavy_ball}, {\em iii}) simulation results in Section \ref{sim} show a significant improvement in the running time of  QNSD with nonzero $\theta$.


\section{Integer NSDP}
\label{integer}


It is immediate to show that the integer NSDP is NP-Hard by reduction from MCND. Recall that the integer NSDP is equivalent to MCND for the special case of single-commodity services. Hence, no $O(\epsilon)$ approximation can in general be computed in sub-exponential time.
After recognizing the difficulty of approximating the integer NSDP,
we now establish key observations on the behavior of QNSD that allows us to add a simple, yet effective, condition on the evolution of QNSD that enables constructing a solution to the integer NSDP of  good practical performance.

We first observe that the QNSD algorithm evolves by allocating an integer number of resources in each iteration. In fact, QNSD solves a max-weight problem in each iteration and allocates either zero or the maximum number of resource units to a single commodity at each cloud network location.
However, the solution in each iteration may significantly violate flow conservation constraints.
On the other hand, the average over the iterations is shown to converge to a feasible, but, in general, fractional solution.
Based on these observations,
we propose C-QNSD (constrained QNSD), a variation of QNSD designed to constrain the solution over the algorithm iterations to satisfy flow conservation across consecutive iterations, such that when the iterative flow solution converges, we can guarantee a feasible solution to the integer NSDP.
C-QNSD works just as QNSD, but with the max-weight problems solved in each iteration 
replaced by the fractional knapsak problems that result from adding the \emph{conditional flow conservation constraints}:
\begin{align}\label{cond}
& \sum_{v\in\delta^{\!+\!}(u)} \!\! f_{uv}^{(d,\phi,i)} (t) \leq  \!\! \sum_{v\in\delta^{\!-\!}(u)} \!\! f_{vu}^{(d,\phi,i)} (t\!-\!1) \qquad \forall d,\phi,i
\end{align}
Specifically, in each iteration of the main procedure, after the queue updates described by \eqref{queue_dynamics}-\eqref{virtual_queue}, the transport and processing decisions are jointly determined as follows:
\begin{itemize}
\item \emph{Transport and processing decisions:}
For each $u\in\mathcal V$: 
\begin{subequations}\label{knapsak}
\begin{align}
&\text{max} &&
\!\!\!\!   \sum_{v\in\delta^{\!+\!}(u)}\sum_{(d,\phi,i)} W^{(d,\phi,i)}_{uv} f^{(d,\phi,i)}_{uv}(t)  - V y_{uv}(t) w_{uv} \notag\\
&\text{s.t.} && \!\!\! \eqref{cond}, \eqref{cap}, \eqref{vars} \notag\\
& \text{where} \notag\\
&&& \!\!\!\!\!\! W_{uv}^{(d,\phi,i)}(t) = \frac{1}{r_{uv}^{tr}} \left( U_u^{(d,\phi,i)}(t) - U_v^{(d,\phi,i)}(t) \right) \notag\\
&&& \qquad\qquad\qquad\qquad\qquad\qquad \forall v\in\delta^{\!+\!}(u)\backslash \{p(u)\} \notag\\
&&& \!\!\!\!\!\! W_{u,p(u)}^{(d,\phi,i)}(t)  = \frac{1}{r_u^{(\phi,i\!+\!1)}} \left( U_u^{(d,\phi,i)}(t) - U_u^{(d,\phi,i\!+\!1)}(t) \right) \notag
\end{align}
\end{subequations}

\end{itemize}

Observe that without the \emph{conditional flow conservation} constraints \eqref{cond}, the problem above can be decoupled into the set of
max-weight problems whose solutions drive the resource allocation and flow rate assignment of QNSD.
When including \eqref{cond}, the solution to the above maximization problem is forced to fill-up the cloud network resource units with multiple commodities, smoothing the evolution towards a feasible integer solution.
In C-QNSD, the above maximization problem is solved via a linear-time heuristic that decouples the problem into a set of fractional knapsacks, one for each neighbor node $v\in\delta^{\!+\!}(u)$ and resource allocation choice $y_{uv}\in\{0,1,\dots,c_{uv}\}$.

As shown in the following section, C-QNSD effectively consolidates service flows into a limited number of active resources. 
Providing some form of performance guarantee is of interest to the authors, but out of the scope of this paper.

\begin{figure}
        \includegraphics[height=3.5cm]{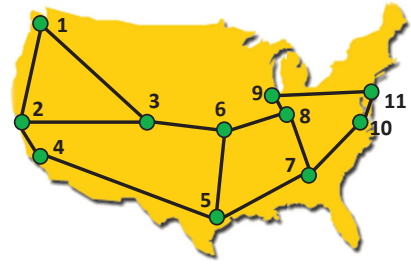}
        \centering{
        \caption{Abilene US continental network topology.}
        \label{abilene}}
\end{figure}

\begin{figure*}[ht]
\centering
\subfigure[]{
\centering \includegraphics[width=6cm,]{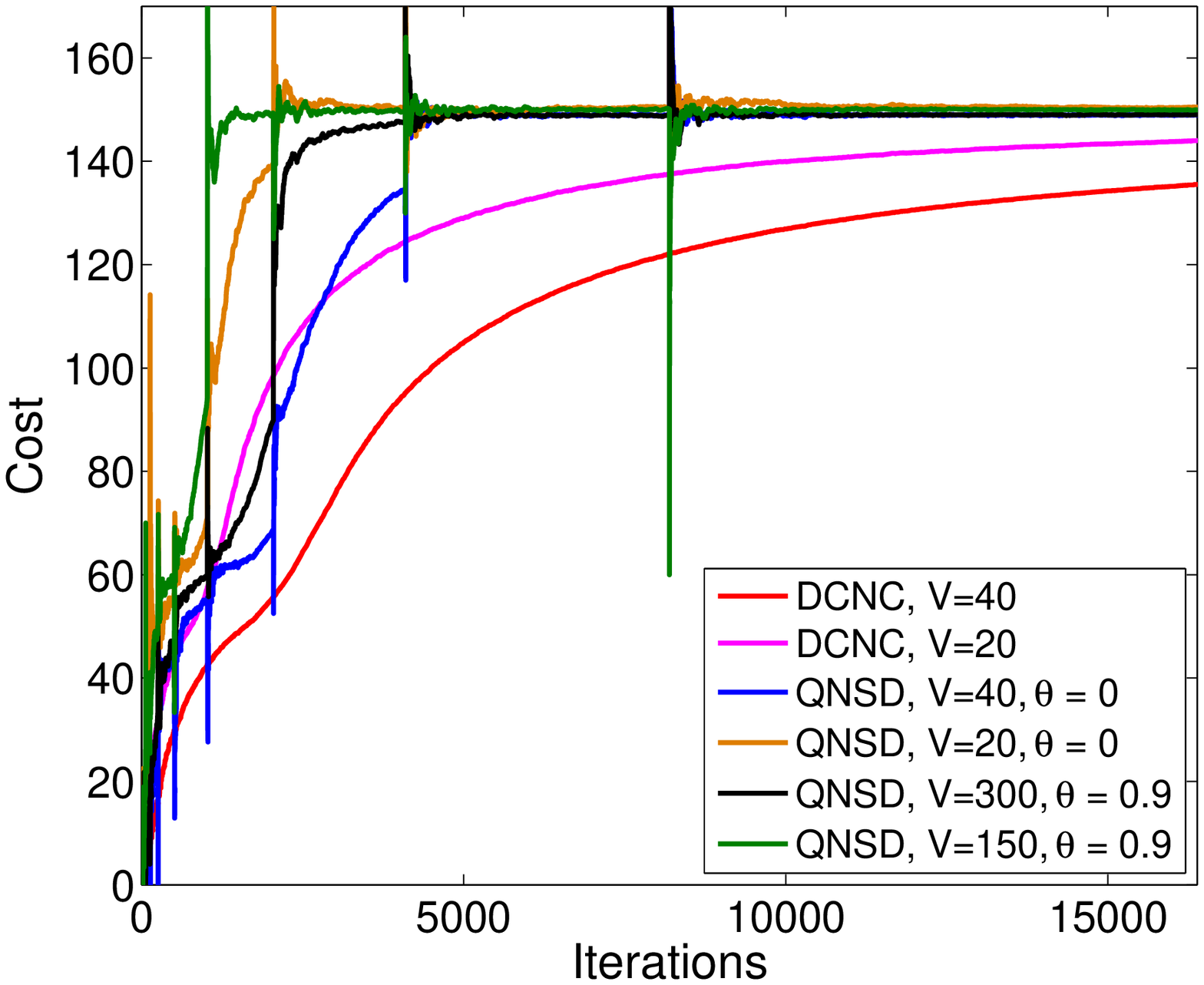}
\label{fig2050}
}
\hspace{-0.7cm}
\subfigure[]{
\centering \includegraphics[width=6cm]{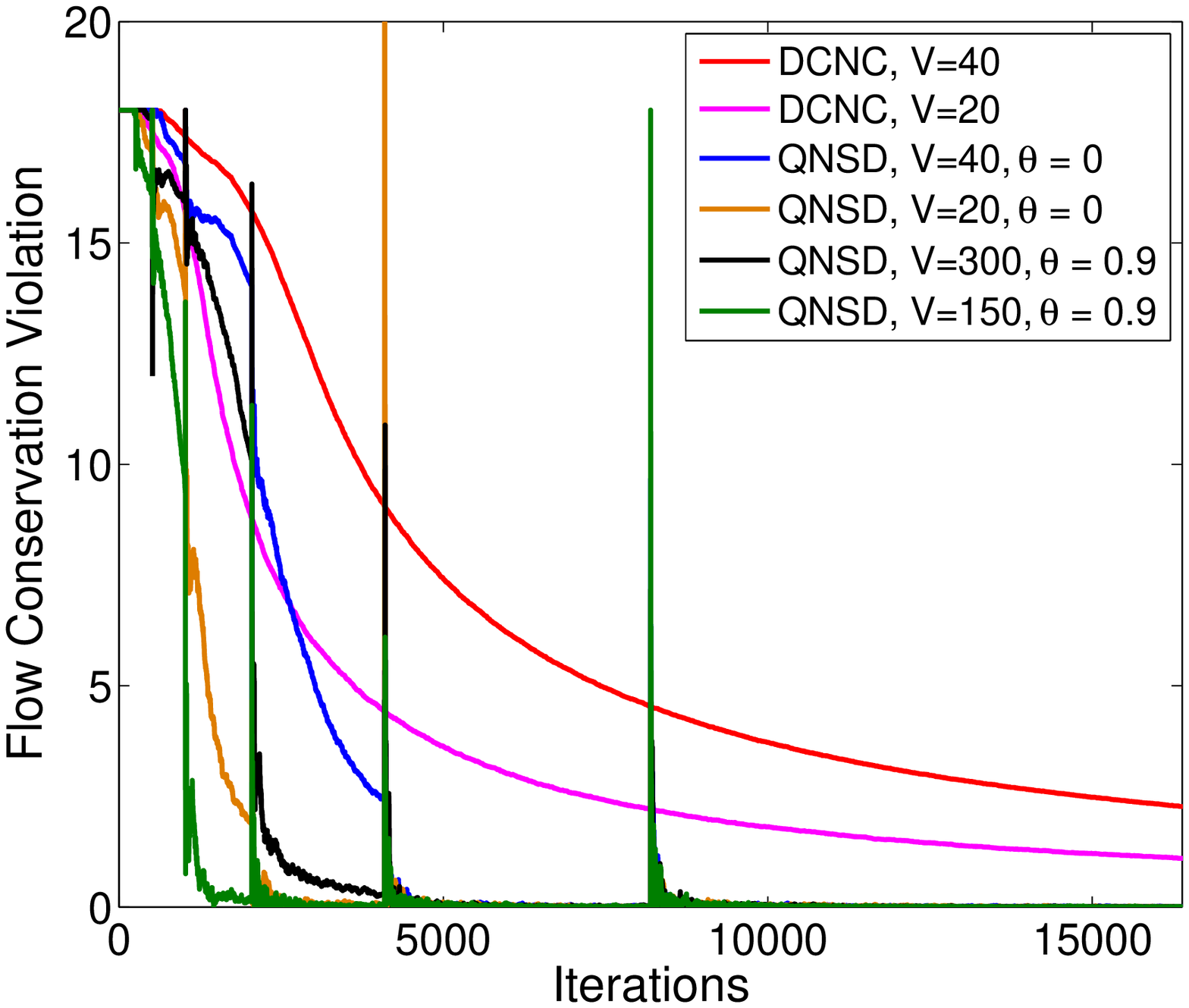}
\label{fig10100}
}
\hspace{-0.7cm}
\subfigure[]{
\centering \includegraphics[width=6cm]{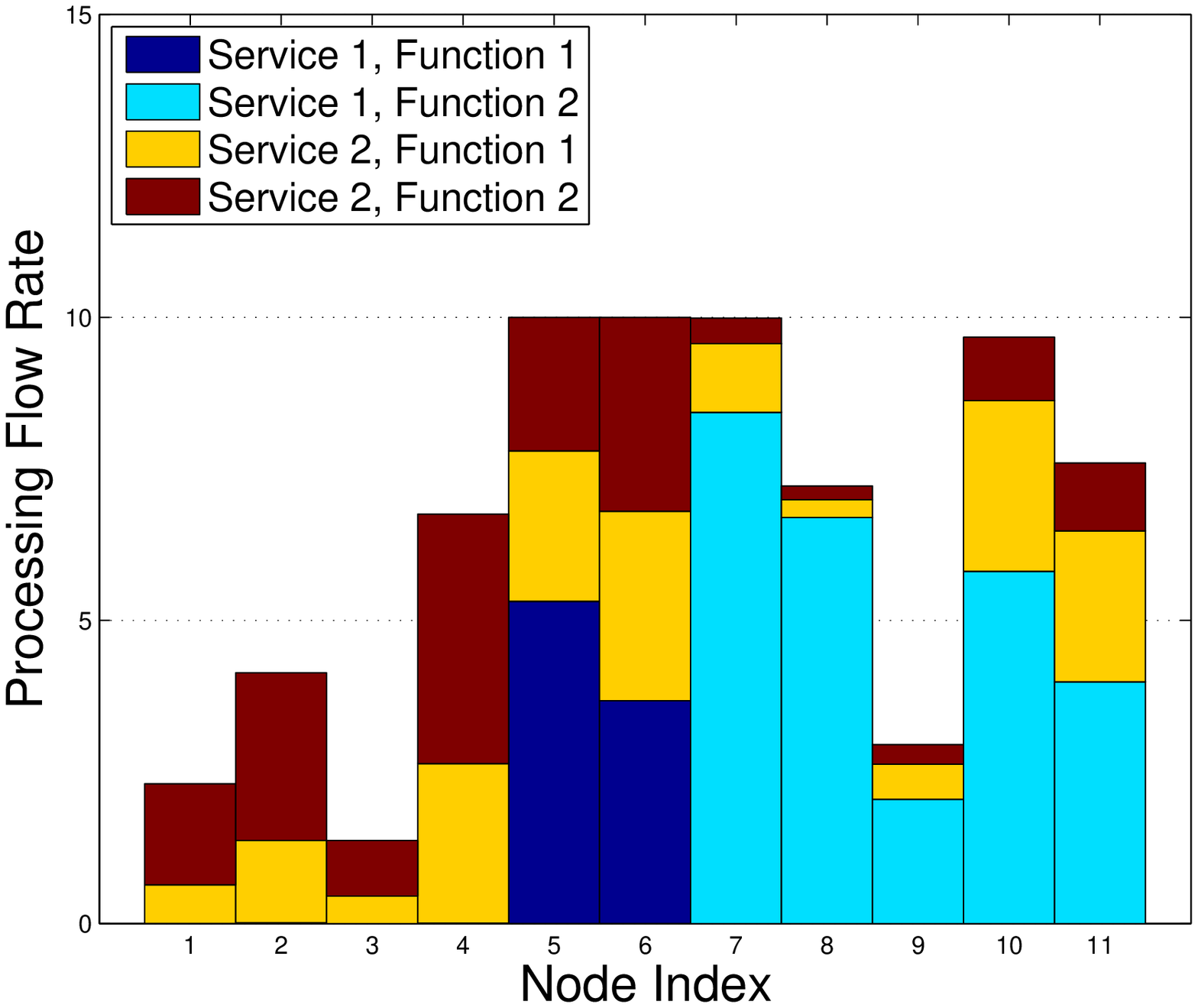}
\label{fig10200}
}
\caption{Performance of QNSD. a)~Evolution of total cost over algorithm iterations; b)~Evolution of flow conservation violation over algorithm iterations; c)~Processing resource allocation distribution across cloud network nodes after running QNSD with $V=300$ and $\theta=0.9$ for 15000 iterations.}
\vspace{-0.5cm}
\label{fig1}
\end{figure*}

\section{Simulation Results}
\label{sim}

In this section, we evaluate the performance of QNSD and C-QNSD in the context of Abilene US continental network, composed of $11$ nodes and $28$ directed links, as illustrated in Fig. \ref{abilene}. We assume each node and link is equipped with $10$ cloud resource units and $10$ network resource units, respectively. The cost per cloud resource unit is set to $1$ at nodes $5$ and $6$, and to $3$ at all other nodes. The cost per network resource unit is set to $1$ for all $28$ links.

\subsection{QNSD}

We first test the performance of QNSD.
We consider a scenario with $2$ services, each composed of $2$ functions. Function $(1,1)$ (service 1, function 1) has resource requirement $1$ resource unit per flow unit; while functions $(1,2)$, $(2,1)$, and $(2,2)$ require $3$, $2$, and $2$ resource units per flow unit, respectively.
We assume resource requirements do not change across cloud locations, and that all links require $1$ network resource unit per flow unit.
There are $6$ clients, represented  by destination nodes $\{1,2,4,7,10,11\}$ in Fig. \ref{abilene}. Each destination node in the west coast $\{1,2,4\}$ has each of the east coast nodes $\{11,10,7\}$ as source nodes, and viceversa,
resulting in a total of $18$ source-destination pairs.
We assume that each east coast client requests service $1$ and each west coast client requests service $2$, and that all input flows have size $1$ flow unit.

Fig. \ref{fig1}a shows the evolution of the cost function over the algorithm iterations. Recall that QNSD exhibits 3 main features: queue-length driven, truncated average computation, and first-order memory. In Fig. \ref{fig1}, we refer to QNSD without truncation and without memory as DCNC, as it resembles the evolution of the dynamic cloud network control algorithm in \cite{dcnc}. We use QNSD with $\theta=0$ to refer to QNSD with truncation, but without memory. And finally, QNSD with $\theta>0$ refers to the QNSD algorithm with both truncation and memory.
It is interesting to observe the improved convergence obtained when progressively including QNSD's key features. Observe how DCNC evolves the slowest and it is not able to reach the optimal objective function value within the $16000$ iterations shown. Decreasing the control parameter $V$ from $40$ to $20$ speeds up the convergence of DCNC, but yields a further-from-optimal approximation, not appreciated in the plot due to the slow convergence of DCNC. In fact, QNSD without truncation and memory can only guarantee a $O(1/\epsilon^2)$ convergence, which is also the convergence speed of the FPTAS for MCF in \cite{Garg07}.
On the other hand, when including truncation, QNSD is able to reach the optimal cost value of $149$ in  around $6000$ iterations, clearly illustrating the faster $O(1/\epsilon)$ convergence. The peaks exhibited by the curves of QNSD with truncation illustrate the reset of the average computations at the beginning of each iteration frame.
Note again that reducing $V$ further speeds up convergence at the expense of slightly increasing the optimality gap. 
Finally, when including the memory feature with a value of $\theta =0.9$, QNSD is able to converge to the optimal solution 
even faster, illustrating the conjectured $O(1/\sqrt\epsilon)$ speed-up from the momentum generated when combining gradient directions (see Section \ref{linear}). Decreasing $V$ again illustrates the speed-up v.s. optimality tradeoff. 

Fig. \ref{fig1}b shows the convergence of the violation of the flow conservation constraints. We can observe a similar behavior as in the cost convergence, with significant speed-ups when progressively adding truncation and memory to QNSD.

Finally, Fig. \ref{fig1}c shows the processing resource allocation distribution across cloud network nodes. As expected, most of the flow processing concentrates on the cheapest nodes $5$ and $6$. Note how function $(1,1)$, which has the lowest processing requirement ($1$ resource unit per flow unit) gets exclusively implemented in nodes $5$ and $6$, as it gets higher priority in QNSD's scheduling decisions. Functions $(2,1)$ and $(2,2)$, which require $2$ resource units per flow unit, share the remaining processing capacity at nodes $5$ and $6$. 
Finally, function $(1,2)$, with resource requirement $3$ resource units per flow unit, and following function $(1,1)$ in the service chain, gets distributed closer to the east coast nodes, destinations for service $1$.



\begin{figure}
        \includegraphics[height=5.3cm]{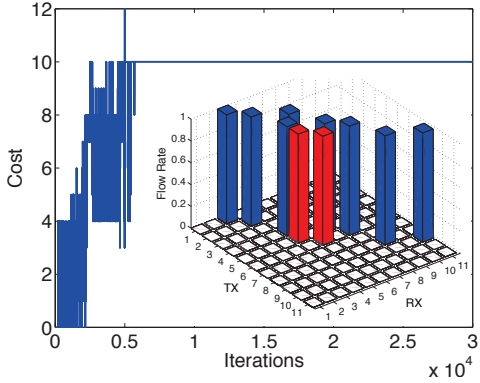}
        \centering{
        \caption{Cost evolution and flow distribution of the C-QNSD solution with input rate $1$ flow unit per client and control parameters $V=1000$, $\theta=0.9$.}
        \label{fig3}}
\end{figure}

\begin{figure}
        \includegraphics[height=5.3cm]{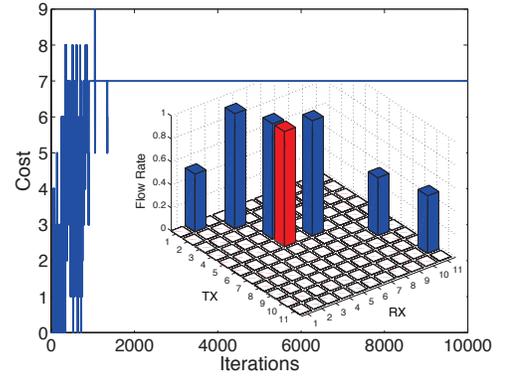}
        \centering{
        \caption{Cost evolution and flow distribution of the C-QNSD solution with input rate $0.5$ flow units per client and control parameters $V=100$, $\theta=0.9$.}
        \label{fig4}}
\end{figure}

\subsection{C-QNSD}

In order to test the performance of C-QNSD, we setup a scenario with $2$ 
$s$-$d$ pairs, $(1, 11)$ and $(2,7)$, both requesting one service composed of one function with resource requirement $1$ cloud resource per flow unit. We simulate the performance of C-QNSD for input rates $1$ flow unit and $0.5$ flow unit per client. Observe from Fig. \ref{fig3} that for input rate $1$, C-QNSD is able to converge to a solution of total cost $10$, in which
each client flow follows the shortest path and where the flow processing happens at nodes $5$ and $6$, respectively.
The 3D bar plot shows the flow distribution over the links (non-diagonal entries in blue) and nodes (diagonal entries in red) in the network.
Observe now the solution for input rate $0.5$ in Fig. \ref{fig4}. 
The flow of $s$-$d$ pair $(1,11)$ is now routed along the longer path $\{1,2,4,5,7,10,11\}$ in order to consolidate as much flow as possible on the activated resources for $s$-$d$ pair $(2,7)$. The flow processing of both services is now consolidated at node $5$, yielding an overall cost of $7$, which is in fact the optimal solution to the integer NSDP in this setting.
Note that if the two client flows were following the shortest path and separately getting processed at nodes $5$ and $6$, without exploiting the available resource consolidation opportunities, the total cost under integer resources would be $10$.

While preliminary,
our results show promise on the combined use of momentum information and conditional flow conservation constraints
for providing efficient solutions to the integer NSDP in practical network settings.

\section{Extensions}
\label{discussion}

While not included in this paper for ease of exposition, our model and algorithms can be extended to include:
\begin{itemize}

\item \emph{Function availability:} Limiting the availability of certain functions to a subset of cloud nodes can be modeled by setting the flow variables associated with function $(\phi,i)$ to zero, $f^{(d,\phi,i)}_{p(u),u}=0$, for all $d$ and for all cloud nodes $u$ in which function $(\phi,i)$ is not available.

\item \emph{Function flow scaling:} Capturing the fact that flows can change size as they go through certain functions can be modeled by letting $\xi^{(\phi,i)}$ denote the number of output flow units per input flow unit of function $(\phi,i)$, and modifying the service chaining constraints \eqref{chain} as $ \xi^{(\phi,i)} f_{p(u),u}^{(d,\phi,i)} = f_{u,p(u)}^{(d,\phi,i-1)}$.

\item \emph{Per-function resource costs:} If the cost function depends on the number of virtual resource units (\eg VMs) instead of on the number of physical resource units (\eg servers), then we can use $y^{(\phi,i)}_{uv}$ and $w^{(\phi,i)}_{uv}$ to denote the number of allocated resource units of function $(\phi,i)$ and the cost per resource unit of function $(\phi,i)$, respectively.

\item \emph{Nonlinear cost function:} In order to capture nonlinear cost effects such as \emph{economies of scale} in which the cost per resource unit decreases with the number of allocated resource units, the cost function can be modified as $\sum_{uv} y_{uv,k} w_{uv,k}$, with $y_{uv,k}\in\{0,1\}$ indicating the allocation of $k$ resource units, and $w_{uv,k}$ denoting the cost associated with the allocation of $k$ resource units. The capacity constraints in \eqref{cap} become $\sum_{(d,\phi,i)} \! f_{uv}^{(d,\phi,i)} r^{(\phi,i+1)}_{uv} \! \leq k\, y_{uv,k} \! \leq  c_{uv} $. 

\end{itemize}

\section{Conclusions}
\label{conclusions}

We have formulated the NFV service distribution problem (NSDP) as a multi-commodity-chain network design problem on a cloud-augmented graph. We have shown that under load-proportional costs, the resulting fractional NSDP becomes a multi-commodity-chain network flow problem that admits optimal polynomial time solutions, and have designed QNSD, a queue-length based iterative algorithm that provides an $O(\epsilon)$ approximation in time $O(1/\epsilon)$.
We further conjectured that, by exploiting the momentum obtained when combining differential queue backlogs across consecutive iterations, QNSD converges in time $O(1/\sqrt{\epsilon})$, and illustrated it via simulations.
We then addressed the case in which resource costs are a function of the integer number of allocated resources. We showed that the integer NSDP is NP-Hard by reduction from MCND and designed C-QNSD, a heuristic algorithm that constrains the evolution of QNSD
to effectively consolidate flows into a limited number of active resources.

\section{Appendix: Proof of Theorem \ref{thm1}}
\label{Appedix:DCNC}


Let
\begin{eqnarray}
\label{eq:dinamicvect}
 \Qm(t+1)  =  \left [\Qm(t) + \Am \fv (t)  \right]^+
\end{eqnarray}
denote the matrix form of the QNSD queuing dynamics given by  \eqref{queue_dynamics}, where  $\fv (t)$, $\Qm(t)$, and $\bf A$ denote the flow vector in iteration $t$, the queue-length vector in iteration $t$, and the matrix of flow coefficients, 
respectively.

Let
\begin{eqnarray}
\label{dnsdp2}
 Z(\Qm)  \triangleq  \inf_{[\yv, \fv]\in \mathcal X }  \left  \{
  V \wv^\dagger \yv+  (\Am \fv )^\dagger \Qm \right \}
  \end{eqnarray}
denote the dual function of the fractional NSDP weighted by the control parameter $V$,
where $\wv$ is the resource cost vector and $\mathcal X$ is the set of feasible solutions to the fractional NSDP.


Let  $\mathcal H^* \triangleq \{ \Qm^* \! :   \, \, \Qm^* =\arg\max_{\Qm}  Z(\Qm) \}$ denote the set of  all optimizers of $ Z(\Qm)$.

Let  $\mathcal H_\gamma \triangleq  \left \{ \Qm:  \sup_{ \Qm^* \in \mathcal H^*} \{ \| \Qm - \Qm^* \| \} \leq \gamma  \right  \}$ denote
the set of queue-length vectors that are at most $\gamma$-away from any point in $\mathcal H^*$. 
We further assume that the value of $\gamma$ is independent of the control parameter $V$.\footnote{While we conjecture that this assumption is not needed to prove Theorem \ref{thm1}, it is technically needed for the type of proof used in this paper.} 


Since $Z(\Qm)$ is a piecewise linear concave  function of  $ \Qm$,
the locally polyhedron property is satisfied,
\ie:
 $\forall \Qm \notin  \mathcal  H_\gamma$ and $\forall \Qm^*\in {\mathcal H}^*$, there  exists a $L_\gamma>0 $ such that
 \begin{eqnarray}
 Z(\Qm^*)-  Z(\Qm) \geq
L_\gamma \| \Qm^*  - \Qm \|.
\label{eq:poly}
 \end{eqnarray}

We then denote the set of queue-length vectors that are $D$-away from $H^*$ as
$$\mathcal H_D \triangleq \left \{ \Qm:  \sup_{ \Qm^* \in \mathcal H^*} \{ \| \Qm - \Qm^* \| \} \leq D  \right  \},$$
with $D \triangleq \max\left \{ \frac{B }{L_\gamma} - \frac{L_\gamma}{4},  \frac{L_\gamma}{2}, \gamma\right\}$
,where $B$ is an upper bound for $\|\Am\fv \|^2$ defined as
\begin{equation}
B \triangleq \sum\limits_u {\left[ {{{\left( {\sum\limits_{v \in {\delta ^ - }\left( u \right)} {\frac{{{c_{vu}}}}{{\mathop {\min }\limits_{(\phi ,i)} r_{vu}^{\left( {\phi ,i} \right)}}}} } \right)}^2} + {{\left( {\sum\limits_{v \in {\delta ^ + }\left( u \right)} {\frac{{{c_{uv}}}}{{\mathop {\min }\limits_{(\phi ,i)} r_{uv}^{\left( {\phi ,i} \right)}}}} } \right)}^2}} \right]}. \nonumber
\end{equation}

We now state and prove the following lemmas, which are used to prove Theorem \ref{thm1}.

\begin{lem}
\label{lemma0}
Under QNSD, if $ \Qm(t) \notin  \mathcal  H_D$, we  have that, for all $\Qm^*\in {\mathcal H}^*$,
$$\| \Qm(t+1) - \Qm^*\|  \leq \| \Qm(t) - \Qm^*\| -\frac{L_\gamma}{2}.$$

\begin{proof}
From \eqref{eq:dinamicvect}, it follows that
\begin{eqnarray}
\| \Qm(t+1) - \Qm^* \|^2 \! \!\!\! &=& \! \!\!\! \| \left [\Qm(t) + \Am \fv (t)   \right]^+   - \Qm^* \|^2 \nonumber \
\\\! \!\!\! &\leq & \! \!\!\!  \| \Qm(t) + \Am \fv (t)       - \Qm^* \|^2  \nonumber  \\
\! \!\!\! &\leq &\! \!\!\!   \| \Qm(t) - \Qm^* \|^2 +  B  \nonumber \\
&& \!\!\!\!\!  + 2 (\Qm(t) - \Qm^*)^\dagger \Am \fv (t)  \label{eq: bound1}
\end{eqnarray}
where  \eqref{eq: bound1} is due to the fact that $\| \Am \fv (t)  \|^2 $ $\leq $ $B$,
as shown in the following:
\begin{align}
&\|\Am\fv \|^2 \nonumber\\
=& \sum\limits_u {\sum\limits_{(d,\phi ,i)} {\left( {\sum\limits_{v \in {\delta ^{ - }}\left( u \right)} {f_{vu}^{(d,\phi ,i)}}  - \sum\limits_{v \in {\delta ^ + }\left( u \right)} {f_{uv}^{(d,\phi ,i)}} } \right)^2} }\nonumber\\
\le &\sum\limits_u {\sum\limits_{(d,\phi ,i)} {\left[ {{{\left( {\sum\limits_{v \in {\delta ^ - }\left( u \right)} {f_{vu}^{(d,\phi ,i)}} } \right)}^2} + {{\left( {\sum\limits_{v \in {\delta ^ + }\left( u \right)} {f_{uv}^{(d,\phi ,i)}} } \right)}^2}} \right]} }  \nonumber\\
\le &\sum\limits_u {\left[ {{{\left( {\sum\limits_{v \in {\delta ^ - }\left( u \right)} {\sum\limits_{(d,\phi ,i)} {f_{vu}^{(d,\phi ,i)}} } } \right)}^2}\!\! +\!\! {{\left( {\sum\limits_{v \in {\delta ^ + }\left( u \right)} {\sum\limits_{(d,\phi ,i)} {f_{uv}^{(d,\phi ,i)}} } } \right)}^2}} \right]} \nonumber\\
\le& \sum\limits_u {\left[ {{{\left( {\sum\limits_{v \in {\delta ^ - }\left( u \right)} {\frac{{{c_{vu}}}}{{\mathop {\min }\limits_{(\phi ,i)} r_{vu}^{\left( {\phi ,i} \right)}}}} } \right)}^2}
+ {{\left( {\sum\limits_{v \in {\delta ^ + }\left( u \right)}{\frac{{{c_{uv}}}}{{\mathop {\min }\limits_{(\phi ,i)} r_{uv}^{\left( {\phi ,i} \right)}}}} } \right)}^2}} \right]}\nonumber\\
\triangleq &B.
\end{align}


From \eqref{dnsdp2},  we have that, for all $\Qm^*$$\in$${\mathcal H}^*$,
\begin{align}
\label{dnsdp3}
Z(\Qm(t)) -  Z(\Qm^*) =& \inf_{[\yv', \fv']\in \mathcal X }  \left  \{
  V \wv^\dagger \yv'+  (\Am \fv'   )^\dagger \Qm(t) \right \}\notag\\
  &- \inf_{[\yv', \fv']\in \mathcal X }  \left  \{
  V \wv^\dagger \yv'+  (\Am \fv'   )^\dagger \Qm^* \right \}\notag\\
  \geq &\ \ V\wv^\dagger \yv(t) + (\Am \fv(t)   )^\dagger \Qm(t)\notag\\
  & - V\wv^\dagger \yv(t) + (\Am \fv(t)   )^\dagger \Qm^*\notag\\
  =&\ \ (\Am \fv (t)   )^\dagger \left( \Qm(t) -\Qm^* \right).
\end{align}

Using \eqref{dnsdp3} in \eqref{eq: bound1}, 
it follows that
\begin{eqnarray}
\label{dnsdp33}
\| \Qm(t+1) - \Qm^* \|^2 \! \!\!\!
 &\leq &\! \!\!\!   \| \Qm(t) - \Qm^* \|^2 +  B  \nonumber \\
&& \!\!\!\!\! - 2 \left ( Z(\Qm^*)-  Z(\Qm(t)) \right).
\end{eqnarray}

From the definition of $\mathcal  H_D$, 
it is immediate to show that for any $\Qm(t) \notin \mathcal  H_D$ and any $\Qm^* \in \mathcal H^*$,
 \begin{eqnarray} \frac{L_\gamma^2}{4}  -  L_\gamma\| \Qm^*- \Qm(t)\| \geq  B   - 2 L_\gamma\| \Qm^*- \Qm(t)\|.
 \label{dnsdp43}
\end{eqnarray}

Now, using  \eqref {eq:poly}  and \eqref{dnsdp43} in \eqref{dnsdp33},
\begin{align}
\| \Qm(t+1) \! - \Qm^* \|^2
 &\! \leq  \| \Qm(t) - \Qm^* \|^2 \! +  B   - 2 L_\gamma\| \Qm^* \! - \Qm(t)\| \notag \\
 &\! \leq \| \Qm(t) - \Qm^* \|^2 \! +  \frac{L_\gamma^2}{4} -  L_\gamma \| \Qm^* \! - \Qm(t)\| \nonumber\\
&\!  =  \left(\| \Qm(t) - \Qm^* \| - \frac{L_\gamma}{2}\right)^2.
\end{align}

Hence, for any $\Qm(t)\notin \mathcal H_D$, the queue length evolution is regulated by
\begin{eqnarray}
\| \Qm(t+1) - \Qm^* \| \! \!\!\!
 &\leq &\! \!\!\!    \| \Qm(t) - \Qm^* \| -  \frac{L_\gamma}{2}.
 \label{reduction}
\end{eqnarray}

\end{proof}

\end{lem}

\begin{lem}
\label{lemma1}
Assuming the Slater condition holds, and letting
\begin{align}
&\mathcal R_D \triangleq \left  \{ \Qm:  \sup_{\Qm^* \in \mathcal H^*} \| \Qm - \Qm^*\|  \leq D + \sqrt{B} \right \}, \notag\\
& \tau_{\mathcal R_D} \triangleq \inf \{ t \geq 0: \Qm(t) \in \mathcal R_D \},
\end{align}
then,  for all $t \geq \tau_{\mathcal R_D}$, $\Qm(t) \in \mathcal R_D $.

\begin{proof}
We follow by induction. By the definition of $\tau_{\mathcal R_D}$,
we have that $\Qm(t)\in \mathcal R_D$ for $t=\tau_{\mathcal R_D}$.
Now, assume $\Qm(t)\in \mathcal R_D$ holds in iteration $t\geq \tau_{\mathcal R_D}$.
Then, in iteration $t+1$, the following holds: if $\Qm(t) \notin \mathcal H_D$, according to Lemma \ref{lemma0}, we have that $\left|\Qm(t+1)-\Qm^*\right|\leq \left|\Qm(t)-\Qm^*\right| - \frac{L_\gamma}{2}\leq D + \sqrt{B} - \frac{L_\gamma}{2}$; on other hand, if $\Qm(\tau_{\mathcal R_D}) \in \mathcal H_D$, it follows that $\left|\Qm(t+1)-\Qm^*\right|\leq \left|\Qm(t)-\Qm^*\right| + \left|\Qm(t+1)-\Qm(t)\right|\leq D + \sqrt{B}$. Hence, in either case, $\Qm(t+1)\in \mathcal R_D$.
\end{proof}
%
%
%
\end{lem}

\bigskip

\begin{lem}
\label{Lemma:Bound}
Letting $\displaystyle h^{\sf max}$$ =$$ \displaystyle \max_{\yv } \left\{\wv^\dagger \yv\right \}$,
if the Slater condition holds, the queue-length vector of QNSD satisfies 
\begin{equation}
\|\Qm(t)\| \leq  \frac{B\left/2\right. +  V h^{\sf max} }{\kappa} + \sqrt{B}, \qquad \forall t >0,
\end{equation}
where $\kappa$ is a positive number satisfying $\|\Qm(t+1)\|^2-\|\Qm(t)\|^2\leq B + Vh^{\sf max} - \kappa\|\Qm(t)\|$.
\begin{proof}
Starting from  the queuing dynamics in \eqref{eq:dinamicvect}, squaring both sides of the equality,
recalling that  $\|\Am  \fv(t)   \|^2 \leq B$, and
 adding $V \wv^\dagger \yv(t)$ on both sides, after algebraic manipulation, we get 
 \begin{eqnarray}
&& \frac{1}{2}\left[ \| \Qm(t+1) \|^2 -  \|\Qm(t)\|^2\right]  +V \wv^\dagger\yv(t)  \nonumber \\
&& \qquad\,\,\,\, \leq   \frac{B}{2}
 + V \wv^\dagger \yv(t) +  \Qm(t)^\dagger \Am \fv (t).
 \label{eq:stronzi}
\end{eqnarray}
As described in Section \ref{QNSD_algorithnm}, QNSD computes the resource and flow vectors $[\yv(t),\fv(t)]$ in iteration $t$ as
\begin{eqnarray}
\label{defalg}
[\yv(t),\fv(t)]  =  \arg\inf_{[\yv', \fv']\in \mathcal X }  \left  \{V \wv^\dagger \yv' +  (\Am \fv'    )^\dagger \Qm(t) \right \},
\label{defalg1}
\end{eqnarray}
which  implies that, 
for any feasible $[\hat \yv, \hat \fv]$,
 \begin{eqnarray}
 &&
 V \wv^\dagger \yv(t)  +    \Qm(t)^\dagger  (\Am \fv (t)   ) \nonumber \\
 &&  \qquad\qquad\quad \leq
V \wv^\dagger  \hat \yv  +    \Qm(t)^\dagger \Am \hat \fv.
\label{rotta}
\end{eqnarray}
Based on the Slater condition, there exists a positive number $\kappa$ and a $[\hat \yv, \hat \fv]$, such that $\Am \hat\fv \le -{\kappa} \bf 1$, and it follows that
\begin{equation}
\label{eq_slater}
V \wv^\dagger \yv(t)  +    \Qm(t)^\dagger \Am \fv (t) \le Vh^{\sf max} -\kappa\|\Qm(t)\|.
\end{equation}

We now consider two cases:

If $\|\Qm(t)\|\geq {{(B\left/2\right. + V{h^{\sf max}})}\left/ \kappa\right.} $, using \eqref{eq_slater} in \eqref{eq:stronzi} yields
\begin{align}
&\frac{1}{2}\left[\| \Qm(t+1) \|^2 -  \|\Qm(t)\|^2\right]\nonumber\\
&\ \ \ \ \leq   \frac{B}{2} + Vh^{\sf max} - V \wv^\dagger\yv(t) -\kappa\|\Qm(t)\|\le 0.
\end{align}

If $\|\Qm(t)\|<{{(B\left/2\right. + V{h^{\sf max}})} \mathord{\left/
 {\vphantom {{(B\left/2\right. + V{h^{\sf max}})} \kappa }} \right.
 \kern-\nulldelimiterspace} \kappa }$, then
 \begin{align}
 \|\Qm(t+1)\|&\le \|\Qm(t)\| + \|\Am\fv(t) \| \nonumber\\
 &\le \frac{B\left/2\right.+Vh^{\sf max}}{\kappa} + \sqrt{B}.
 \end{align}

Hence, in either case, we can upper bound $\|\Qm(t)\|$ by ${{(B\left/2\right. + V{h^{\sf max}})} \mathord{\left/
 {\vphantom {{(B\left/2\right. + V{h^{\sf max}})} \kappa }} \right.
 \kern-\nulldelimiterspace} \kappa } + \sqrt B $ for all $t>0$.

\end{proof}

\end{lem}

\bigskip
We now proceed to  prove Theorem \ref{thm1}.
Using \eqref{rotta} and evaluating the right-hand side of \eqref{eq:stronzi}
at the optimal solution  $[\yv^{\sf opt}$, $ \fv^{\sf opt}]$, we obtain
\begin{eqnarray}
\frac{1}{2}  \!\!\!\!\!\!\!&&  \!\!\!\!\!\!
\left (\| \Qm(t+1) \|^2  -  \|\Qm(t)\|^2 \right)   +V \wv^\dagger \yv  \nonumber \\
&&\leq \frac{B}{2} +  V \wv^\dagger \yv^{\sf opt}   +  \Qm(t)^\dagger \Am \fv^{\sf opt}(t)
\nonumber \\
&&=
 \frac{B}{2} +  V \wv^\dagger \yv^{\sf opt},
 \label{lyap46}
\end{eqnarray}
where the last equality follows from $\Am \fv^{\sf opt}(t)  = \bf 0$.
Now, denoting $h^{\sf opt}=  \wv^\dagger \yv^{\sf opt}$,
from  \eqref{lyap46}, we have that
\begin{eqnarray}
\wv^\dagger \yv(t)  - h^{\sf opt}\leq
 \frac{B}{2V }  -
  \frac{1}{2V }
\left (\| \Qm(t+1) \|^2  -  \|\Qm(t)\|^2 \right).
 \label{lyap00}
\end{eqnarray}
Next, taking the average  of  \eqref{lyap00}  over the $\ztr$-th iteration frame $[t_\ztr, t_\ztr+\Delta_{t_\ztr}]$,
we obtain
\begin{eqnarray}
\frac{1}{\Delta_{t_\ztr}+1} \!\!\!\!\!\!&& \!\!\!\!\!\!\sum_{\tau=t_{\ztr}}^{t_\ztr+\Delta_{t_\ztr}}\wv^\dagger \yv(\tau)  -  h^{\sf opt}
\nonumber \\
 &&\leq   \frac{B}{2V }  +
  \frac{\left ( \|\Qm(t_\ztr)\|^2 -\| \Qm(t_\ztr+\Delta_{t_\ztr}) \|^2   \right)}{2V(\Delta_{t_\ztr}+1)  }
 \nonumber
 \label{lyap}
\end{eqnarray}
Next, we use the following two properties:
\begin{itemize}
\item {\bf R1:}   $\tau_{\mathcal R_D}=O\left({h^{\sf max}V}\left/{L_\gamma\kappa}\right.\right)=O(mV)$.
\item {\bf R2:}  if $t_\ztr\geq\tau_{\mathcal R_D}$ then $\|\Qm(t_\ztr)\|^2 -\| \Qm(t_\ztr+\Delta_{t_\ztr}) \|^2=O ((D+\sqrt{B})h^{\sf max}V\left/L_{\gamma}\kappa\right.)=O(m^2V)$.
\end {itemize}

To prove {\bf R1}, according to Lemma \ref{lemma0}, we first have $\tau_{\mathcal R_D}\le \left\lceil {{{2\left\| {\Qm^*-\Qm(0)} \right\|}}\left/{{{L_\gamma }}}\right.} \right\rceil=\left\lceil {{{2\left\| {\Qm^*} \right\|}}\left/{{{L_\gamma }}}\right.} \right\rceil$. Moreover, from Lemma \ref{Lemma:Bound}, note that
\begin{align}
\|\Qm^*\|&\le \|\Qm(\tau_{\mathcal R_D})\| + \|\Qm^*-\Qm(\tau_{\mathcal R_D})\| \nonumber\\
&\le \frac{{B\left/2\right. + V{h^{\sf max}}}}{\kappa } + \sqrt{B} + D.
\end{align}
Hence, we have
\begin{equation}
\label{eq: tau}
\tau_{\mathcal R_D}\le \left\lceil \frac{2}{L_{\gamma}}\left(\frac{{B\left/2\right. + V{h^{\sf max}}}}{\kappa } + \sqrt{B} + D\ \right)\right\rceil,
\end{equation}
and therefore we can write $\tau_{\mathcal R_D}$ as
\begin{equation}
\label{tau_R_D}
\tau_{\mathcal R_D} = O\left(\frac{h^{\sf max}V}{L_\gamma\kappa}\right) = O(mV).
\end{equation}
The above relation holds true because $h^{\sf max}\le mh_0^{\sf max}=O(m)$, where $h_0^{\sf max}$ is the maximum per-node-cost.

To prove {\bf R2}, according to Lemma \ref{lemma1}, if $t_j\ge \tau_{\mathcal R_D}$, by using \eqref{eq: tau}, then we have
\begin{align}
&\|\Qm(t_j)\|^2 - \|\Qm(t_j+\Delta_{t_j})\|^2\nonumber\\
=&\|\Qm(t_j)-\Qm^*\|^2 - \|\Qm(t_j+\Delta_{t_j})-\Qm^*\|^2 \nonumber\\
&+ 2\Qm^{*\dagger}\left[\Qm(t_j)-\Qm(t_j+\Delta_j)\right]\nonumber\\
\le&\|\Qm(t_j)-\Qm^*\|^2 + 2\|\Qm^*\|\cdot\|\Qm(t_j)-\Qm(t_j+\Delta_j)\|\nonumber\\
\le&\left(D+\sqrt{B}\right)^2 + 2\|\Qm^*\|\left(D+\sqrt{B}\right)\nonumber\\
=&O\left(\frac{\left(D+\sqrt{B}\right)h^{\sf max}V}{\kappa}\right) \nonumber\\
=& O(m^2V),
\end{align}
where the last equality is due to $h^{\sf max}=O(m)$, $D=O(B)$, and
\begin{align}
B \le & 2\sum\limits_u {\left[ {{{\sum\limits_{v \in {\delta ^ - }\left( u \right)} {\left( {\frac{{{c_{vu}}}}{{\mathop {\min }\limits_{(\phi ,i)} r_{vu}^{\left( {\phi ,i} \right)}}}} \right)} }^2}\!\! +\!\! {{\sum\limits_{v \in {\delta ^ + }\left( u \right)} {\left( {\frac{{{c_{uv}}}}{{\mathop {\min }\limits_{(\phi ,i)} r_{uv}^{\left( {\phi ,i} \right)}}}} \right)} }^2}} \right]}\nonumber\\
\le &2m\mathop {\max }\limits_{(u,v) \in \mathcal E^a} \frac{{c_{uv}^2}}{{\mathop {\min }\limits_{(\phi ,i)} {{\left( {r_{uv}^{(\phi ,i)}} \right)}^2}}}=O(m).
\end{align}

Using  {\bf R1}  and {\bf R2}, and letting $V=1\left/\epsilon\right.$ and $\Delta_{t_j} \ge \lceil mV-1 \rceil$, it follows from \eqref{lyap00} that, for all $t_\ztr \ge  \tau_{\mathcal R_D}$,
\begin{align}
\label{lyap}
&\wv^\dagger  \left( \sum_{\tau=t_\ztr}^{t_\ztr +\Delta_{t_\ztr}}  \frac{\yv(\tau) }{\Delta_{t_\ztr}+1} \right ) -  h^{\sf opt}\nonumber\\
\le&\frac{B}{2V} + \frac{{{h^{\sf max}}\left( {D + \sqrt B } \right)}}{{\kappa \left( {{\Delta _j} + 1} \right)}} + \frac{{3\kappa {{\left( {D + \sqrt B } \right)}^2} + B\left( {D + \sqrt B } \right)}}{{2\kappa V\left( {{\Delta _j} + 1} \right)}}\nonumber\\
=&O\left(\frac{B}{2V} + \frac{{{h^{\sf max}}\left( {D + \sqrt B } \right)}}{{m\kappa V}}\right)=O\left(\frac{m}{V}\right)=O(m\epsilon).
\end{align}

Observing that $  \bar y_{uv} = \sum_{\tau=t_\ztr}^{t_\ztr+\Delta_{t_\ztr}} \frac{y_{uv}(\tau) }{\Delta_{t_\ztr}+1}$,
Eq. \eqref{eq:thm1} in Theorem \ref{thm1}  follows immediately.
In order to prove Eq. \eqref{eq:thm2},  starting from  \eqref{eq:dinamicvect}, we get
 \begin{eqnarray}
 \Qm(t+1)  
 & \geq & \Qm(t) + \Am \fv (t) ,
\end{eqnarray}
from which it follows that, for all $t _\ztr \!\ge\! \tau_{\mathcal R_D}$,
 \begin{eqnarray}
 \sum_{\tau=t_\ztr}^{t_\ztr+\Delta_{t_\ztr}}   \frac{ \Am \fv (t) }{\Delta_{t_\ztr}+1}   &\leq &
  \frac{\Qm(t_\ztr+\Delta_{t_\ztr}) - \Qm(t_\ztr) }{\Delta_{t_\ztr}+1}.
   \label{lyap112}
   \end{eqnarray}
Finally, the  $\{u,(d,\phi,i)\}$-th element of the vectors in  \eqref{lyap112} satisfies
\begin{align}
\label{flow_conserv}
&\ \ \ \  \sum_{v\in\delta^{\!-\!}(u)} \!   \bar f_{vu}^{(d,\phi,i)} - \! \sum_{v\in\delta^{\!+\!}(u)} \!  \bar  f_{uv}^{(d,\phi,i)}     \nonumber \\
&\leq\frac{\left|Q_u^{(d,\phi,i)}(t_\ztr+\Delta_{t_\ztr}) -Q_u^* \right| + \left| Q_u^{(d,\phi,i)}(t_\ztr)-Q_u^*\right| }{\Delta_{t_\ztr}+1}\nonumber\\
&\le \frac{2 (D + \sqrt{B}) }{mV} = O\left(\frac{1}{V}\right)= O(\epsilon),
\end{align}
where
$$ \bar f_{vu}^{(d,\phi,i)} = \sum_{\tau=t_{\ztr}}^{t_\ztr+\Delta_{t_\ztr}}     \frac{f_{vu}^{(d,\phi,i)}(\tau)} {\Delta_{t_\ztr}+1},
\quad    \bar f_{uv}^{(d,\phi,i)}  = \sum_{\tau=t_\ztr}^{t_\ztr+\Delta_{t_\ztr}}     \frac{f_{uv}^{(d,\phi,i)}(\tau)}{\Delta_{t_\ztr}+1}.$$

Note that $\tau_{\mathcal R_D} = O(mV)=O(m\left/\epsilon\right.)$ due to \eqref{tau_R_D}, and $\Delta_{t_j} \ge \left\lceil {mV-1} \right\rceil  =O(m\left/\epsilon\right.)$. Thus, from \eqref{lyap} and \eqref{flow_conserv}, $\forall t\ge T(\epsilon) \triangleq t_j + \Delta_{t_j}+1\ge \tau_{\mathcal R_D} + mV = O(m/\epsilon)$, the following relations hold:
\begin{align}
& \! \sum_{(u,v)\in\mathcal E^a} \! w_{uv} \, \bar y_{uv}  \leq  h^{ \sf opt} + O(m\epsilon) \label{eq:thm2_finally1}\\
& \! \sum_{v\in\delta^{\!-\!}(u)} \! \bar f_{vu}^{(d,\phi,i)} - \! \sum_{v\in\delta^{\!+\!}(u)} \!
\bar f_{uv}^{(d,\phi,i)} \leq O(\epsilon)  \qquad\,\,\, \forall u, d, \phi, i \label{eq:thm2_finally2}
\end{align}
Note that since QNSD chooses the $j$-th iteration frame as $\ztr=\lfloor \log_2 t \rfloor$, $t_j=2^\ztr$, and  $\Delta_{t_\ztr}=2^\ztr$,
there exists a $j^*$ such that $2^{j^*} \in \left[ \max\{\tau_{{\mathcal R}_D}, mV+1 \}, 2\max\{\tau_{{\mathcal R}_D}, mV+1 \} \right)$,
and hence Eq. \eqref{eq:thm2_finally1} and \eqref{eq:thm2_finally2} hold.


In addition, since $f_{uv}^{(d,\phi,i)}(t)$ and $y_{uv}(t)$ satisfy \eqref{chain}-\eqref{vars} in each iteration of QNSD, so does the final solution $\{\bar f_{uv}^{(d,\phi,i)}, \bar y_{uv}\}$, which concludes the proof of Theorem \ref{thm1}.

\hfill $\blacksquare$


\section*{Acknowledgements}
This work was supported in part by the US National Science Foundation.

\end{document}